\documentclass[letterpaper, 10 pt, conference]{ieeeconf}  

\IEEEoverridecommandlockouts                              

\usepackage{graphics} 
\usepackage{epsfig} 
\usepackage{times} 
\usepackage{amsmath} 
\usepackage{amssymb}  
\usepackage{cite}
\usepackage{epstopdf}
\usepackage{subfigure,color,balance}
\usepackage{verbatim}
\usepackage{cases}
\usepackage{enumerate}
\usepackage{balance}

\usepackage[colorlinks=true,      
linkcolor=black,      
citecolor=black,      
filecolor=black,      
urlcolor=blue]{hyperref}

\newtheorem{definition}{Definition}
\newtheorem{theorem}{Theorem}

\newtheorem{remark}{Remark}
\newtheorem{lemma}{Lemma}
\newtheorem{corollary}{Corollary}

\def\0{{\bf 0}}
\def\1{{\bf 1}}

\def\eg{{\em e.g.}}
\def\ie{{\em i.e.}}

\title{\LARGE \bf
Leading Cruise Control in Mixed Traffic Flow
}

\author{Jiawei Wang$^1$, Yang Zheng$^{2*}$, Chaoyi Chen$^1$, Qing Xu$^1$ and Keqiang Li$^{1*}$
\thanks{This work is supported by National Key R\&D Program of China with 2019YFB160080, Tsinghua University-Didi Joint Research Center for Future Mobility, and the Joint Laboratory for Internet of Vehicles, Ministry of Education-China Mobile Communications Corporation. All correspondence should be sent to Y.~Zheng and K.~Li.}
\thanks{$^{1}$J.~Wang, C.~Chen, Q.~Xu and K.~Li are with the School of Vehicle and Mobility, Tsinghua University, Beijing, China, and with Tsinghua University-Didi Joint Research Center for Future Mobility, Beijing, China. (\{wang-jw18,chency19\}@mails.tsinghua.edu.cn, \{qingxu,likq\}@tsinghua.edu.cn).}%
\thanks{$^{2}$Y. Zheng is with the SEAS and CGBC at Harvard University, Cambridge, USA. ({zhengy@g.harvard.edu}).}%
}

\begin{document}

\maketitle
\thispagestyle{empty}
\pagestyle{empty}

\begin{abstract}

Vehicle-to-vehicle (V2V) communications have a great potential to improve traffic system performance. Most existing work of connected and autonomous vehicles (CAVs) focused on adaptation to downstream traffic conditions, neglecting the impact of CAVs' behaviors on upstream traffic flow. In this paper, we introduce a notion of Leading Cruise Control (LCC) that retains the basic car-following operation and explicitly considers the influence of the CAV's actions on the vehicles behind. We first present a detailed modeling process for LCC. Then, rigorous controllability analysis verifies the feasibility of exploiting the CAV as a leader to actively lead the motion of its following vehicles. Besides, the head-to-tail transfer function is derived for LCC under adequate employment of V2V connectivity. Numerical studies confirm the potential of LCC to strengthen the capability of CAVs in suppressing traffic instabilities and smoothing traffic flow.

\end{abstract}

\section{Introduction}

Vehicle-to-vehicle (V2V) and vehicle-to-infrastructure (V2I) communications have provided new opportunities for better traffic mobility and smoother traffic flow~\cite{contreras2017internet}. As a prevailing extension of Adaptive Cruise Control (ACC)~\cite{vahidi2003research}, Cooperative Adaptive Cruise Control (CACC) organizes multiple adjacent connected and autonomous vehicles (CAVs) into a platoon and has shown a great potential in reducing inter-vehicle distances and mitigating traffic perturbations~\cite{li2017dynamical} (see Fig.~\ref{Fig:SystemSchematic}(a)). Considering that a practical scenario in the near future is the coexistence of human-driven vehicles (HDVs) and CAVs, extending ACC to mixed traffic scenarios has received significant attention. One typical framework is the notion of Connected Cruise Control (CCC), which explicitly considers HDVs' car-following dynamics and determines control strategies for the CAV at the tail by monitoring the motion of multiple HDVs ahead~\cite{orosz2016connected} (see Fig.~\ref{Fig:SystemSchematic}(b)).

Similar to the decision-making characteristics of human drivers \cite{treiber2013traffic}, CCC-type controllers focus on monitoring the downstream traffic conditions consisting of the vehicles ahead, and aim at achieving a better car-following behavior; see, \eg, \cite{orosz2016connected,di2019cooperative,jin2014dynamics,zhou2020stabilizing}. It is worth noting that due to this kind of front-to-rear reaction dynamics of human drivers, the behavior of one individual vehicle will simultaneously have a significant impact on the upstream traffic flow containing the vehicles behind. One example from the negative side is that small perturbations might be amplified and grow into a stop-and-go pattern, causing great loss of efficiency and high risks of accidents \cite{sugiyama2008traffic}. The topic of how to enable CAVs to dissipate the perturbations coming from front has already gained widespread attention \cite{di2019cooperative,jin2014dynamics,zhou2020stabilizing}; however, the subsequent influence of the perturbations on the upstream traffic flow behind the CAV has not been clearly addressed. If such influence is explicitly incorporated into CAVs' controller design, the potential of CAVs could be further explored.

Along this direction, this paper introduces a new notion of Leading Cruise Control (LCC); see Fig.~\ref{Fig:SystemSchematic}(c) for illustration. Originating from ACC and CCC, LCC retains the basic car-following operation that adapts to the motion of the preceding vehicles. Meanwhile, LCC explicitly considers the impact of the behavior of an individual vehicle on its surrounding traffic, especially those HDVs following behind. Particularly, the CAV acts as a leader to actively lead the motion of the following vehicles without changing their natural behavior, and aims to improve the performance of the entire traffic flow. A similar idea is the recently proposed notion of Lagrangian control of traffic flow \cite{stern2018dissipation}, where CAVs are utilized as mobile actuators for traffic control. Existing research has revealed the potential of one single CAV in stabilizing a closed ring-road traffic system \cite{zheng2020smoothing,wang2020controllability,cui2017stabilizing}. Since the impact of the CAV propagates backwards, \ie, upstream in traffic flow, LCC generalizes these ring-road results to common open straight road scenarios, and enables the CAV to stabilize its following vehicles and improve their performance. This possibility has been recently investigated from a macroscopic view, where the CAV is employed as a moving bottleneck to influence the upstream traffic conditions \cite{piacentini2018traffic}. Unlike \cite{piacentini2018traffic}, LCC focuses on the microscopic car-following dynamics, which is more feasible for practical use in individual vehicles, as this has been achieved in ACC and CCC.

\begin{figure*}[t]
	\vspace{1mm}
	\hspace{1.5mm}
	\includegraphics[scale=0.45]{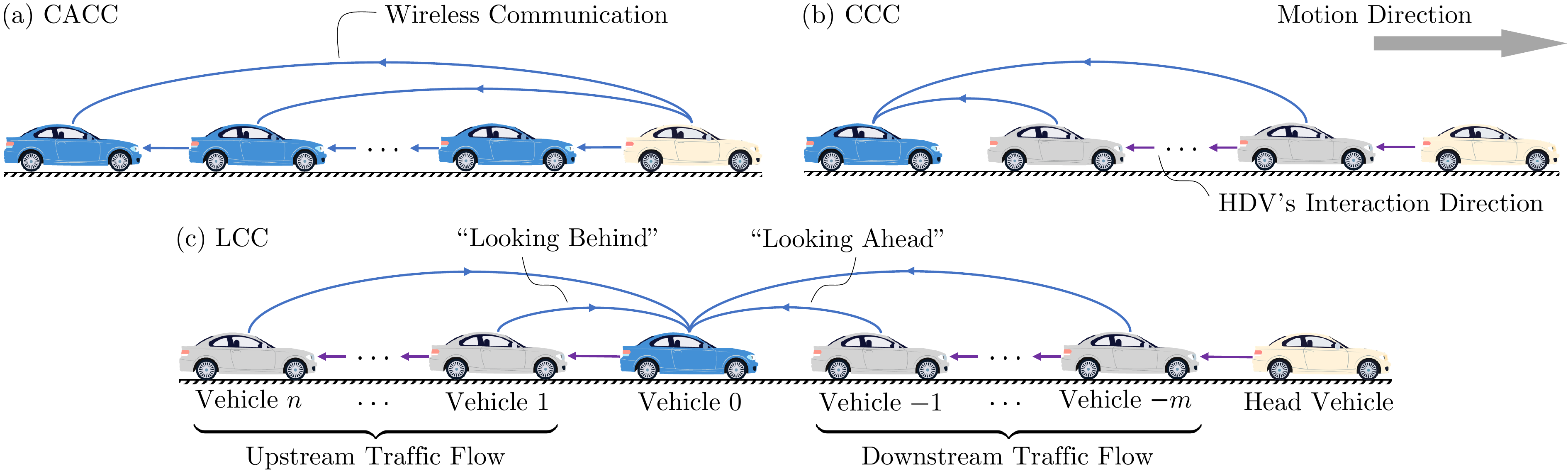}
	\vspace{-2mm}
	\caption{Schematic of different control frameworks for CAVs. The blue arrows represent the communication topology of the CAV, while the purple arrows illustrate the interaction direction in HDVs' dynamics. The blue vehicles, gray vehicles and yellow vehicles represent CAVs, HDVs and the head vehicle, respectively. (a) In a CACC platoon, multiple CAVs are controlled to follow a designated head vehicle. Here, we demonstrate a typical communication topology named predecessor-leader following (see~\cite{zheng2016stability} for other typical topologies). (b) CCC focuses on a mixed traffic scenario, where the CAV receives the information from multiple HDVs ahead~\cite{orosz2016connected}. (c) In LCC, the CAV explicitly considers the dynamics of both the vehicles ahead and the vehicles behind.}
	\label{Fig:SystemSchematic}
	\vspace{-5mm}
\end{figure*}

One significant feature of LCC is the explicit consideration of an individual vehicle as both a leader and a follower in traffic flow. It is known that in complex networks or multi-agent systems, there are two main topics: 1) the control of follower agents, targeting at tracking a prescribed trajectory of leader agents \cite{ding2013network,ni2010leader}; 2) the control of leader agents that act as control inputs to achieve desired performance for the entire system \cite{jafari2011leader,clark2013supermodular}. In multi-vehicle systems, however, most existing research, \eg, ACC, CACC and CCC, focuses on how to improve the behavior of the CAV as a follower agent, but neglects another role of the CAV, which is naturally a leader agent with regard to the vehicles behind, especially in mixed traffic flow. Another distinction of LCC from previous research is the adequate employment of V2V connectivity, where the information of both the HDVs ahead and those behind is utilized for controller design. This communication topology can be informally called as both ``looking ahead'' and ``looking behind'', as shown in Fig.~\ref{Fig:SystemSchematic}(c).

In this paper, we investigate the fundamental properties of the LCC framework, and our contributions are as follows.
\begin{itemize}
\item   We introduce a new notion of Leading Cruise Control. Based on linearized car-following models, a dynamical modeling framework of LCC is presented, and two special cases are discussed, which cover two fundamental driving behaviors of individual vehicles in traffic flow: car-following and free-driving \cite{treiber2013traffic}.
\item   We prove that the motion of the vehicles behind is completely controllable by the control input of the CAV under a very mild condition. These results generalize the previous stabilizability results in a closed ring-road system \cite{zheng2020smoothing,wang2020controllability,cui2017stabilizing}, and verify the possibility of traffic control via CAVs on the common open road.
\item   We investigate the head-to-tail string stability of the mixed traffic flow under the LCC framework. Numerical studies reveal that the CAV has more choices of string stable feedback policies, as well as a greater capability to suppress traffic instabilities, after ``looking behind'' compared with ``looking ahead'' only \cite{di2019cooperative,jin2014dynamics,zhou2020stabilizing}.

\end{itemize}

The remainder of this paper is organized as follows. Section \ref{Sec:2} introduces the modeling process of LCC. Section~\ref{Sec:3} presents the controllability analysis, and Section \ref{Sec:4} investigates the head-to-tail string stability. Finally, Section~\ref{Sec:5} concludes this paper.

\section{Theoretical Modeling Framework for LCC}
\label{Sec:2}

In this section, we first introduce the longitudinal dynamics of HDVs' car-following behavior. Then, we present the modeling framework of the general LCC system, as well as two special cases.

Consider an open single-lane setup, as shown in Fig.~\ref{Fig:SystemSchematic}(c). We index the CAV as vehicle $0$, and define $\mathcal{F}=\{1,2,\ldots,n\}$ and $\mathcal{P}=\{-1,-2,\ldots,-m\}$ as the set of the indexes of the following vehicles and the preceding vehicles that are connected to the CAV, respectively. The position, velocity and acceleration of vehicle $i$ is denoted as $p_i$, $v_i$ and $a_i$, respectively. The spacing of vehicle $i$ from its preceding vehicle is defined as $s_i=p_{i-1}-p_i$. There also exists a vehicle at the head of this series of vehicles, whose information is not received by the CAV, and its velocity is represented as $v_\mathrm{h}$. Without loss of generality, the vehicle length is ignored.

\subsection{Longitudinal Dynamics of HDVs}
We first model the longitudinal car-following dynamics of individual HDVs. There are many continuous-time models in the literature, \eg, optimal velocity model (OVM) \cite{bando1995dynamical}, intelligent driver model (IDM) \cite{Treiber2000Congested} and their variants. Most of them can be written as the following form \cite{jin2017optimal,cui2017stabilizing}
\begin{equation}\label{Eq:HDVModel}
	\dot{v}_i(t)=F\left(s_i(t),\dot{s}_i(t),v_i(t)\right),
\end{equation}
which means that the acceleration of vehicle $i$ depends on the relative distance, relative velocity $\dot{s}_i(t)=v_{i-1}(t)-v_i(t)$, and its own velocity. In equilibrium traffic state, each vehicle moves with the same equilibrium velocity $v^*$ and a corresponding equilibrium spacing $s^*$, which should satisfy
\begin{equation} \label{Eq:Equilibrium}
F\left(s^*,0,v^*\right)=0.
\end{equation}
Assuming that traffic flow is in near-equilibrium conditions, we define the error state of vehicle $i$ as
$
\tilde{s}_i(t)=s_i(t)-s^*,\tilde{v}_i(t)=v_i(t)-v^*.
$
Then we can derive a linearized second-order model for each HDV ($i\in \mathcal{F}\cup\mathcal{P}$) by using~\eqref{Eq:Equilibrium} and applying the first-order Taylor expansion to \eqref{Eq:HDVModel},
\begin{equation}\label{Eq:LinearHDVModel}
\begin{cases}
\dot{\tilde{s}}_i(t)=\tilde{v}_{i-1}(t)-\tilde{v}_i(t),\\
\dot{\tilde{v}}_i(t)=\alpha_{1}\tilde{s}_i(t)-\alpha_{2}\tilde{v}_i(t)+\alpha_{3}\tilde{v}_{i-1}(t),\\
\end{cases}
\end{equation}
with $\alpha_{1} = \frac{\partial F}{\partial s}, \alpha_{2} = \frac{\partial F}{\partial \dot{s}} - \frac{\partial F}{\partial v}, \alpha_{3} = \frac{\partial F}{\partial \dot{s}}$ evaluated at the equilibrium state ($s^*,v^*$). To reflect the real driving behavior, we have $\alpha_{1}>0, \alpha_{2}>\alpha_{3}>0$ \cite{cui2017stabilizing,jin2017optimal}.

One prevailing specific HDV model is the following nonlinear OVM model \cite{bando1995dynamical}, where the explicit expression for \eqref{Eq:HDVModel} becomes
\begin{equation} \label{Eq:OVMmodel}
F(\cdot)=\alpha\left(V\left(s_i(t)\right)-v_i(t)\right)+\beta\dot{s}_i(t).
\end{equation}
$V(s)$ is the spacing-dependent desired velocity of the human driver, typically given by a continuous piecewise function
\begin{equation}
V(s)=\begin{cases}
0, &s\le s_{\mathrm{st}};\\
\frac{v_{\max }}{2}\left(1-\cos (\pi\frac{s-s_{\mathrm{st}}}{s_{\mathrm{go}}-s_{\mathrm{st}}})\right), &s_{\mathrm{st}}<s<s_{\mathrm{go}};\\
v_{\max}, &s\ge s_{\mathrm{go}}.
\end{cases}
\end{equation}
This means that the desired velocity becomes zero for a small spacing $s_{\mathrm{st}}$, and reaches a maximum value $v_{\max}$ for a large spacing $s_{\mathrm{go}}$. When $s_{\mathrm{st}}<s<s_{\mathrm{go}}$, the desired velocity increases monotonically.
Upon using the OVM model \eqref{Eq:OVMmodel}, the specific expression for the coefficients in the linearized dynamics \eqref{Eq:LinearHDVModel} is given by $\alpha_1 = \alpha \dot{V}(s^*),\alpha_2=\alpha+\beta,\alpha_3=\beta$, where $\dot{V}(s^*)$ denotes the derivative of $V(s)$ at $s^*$.

\subsection{Modeling the General LCC System}

 The longitudinal dynamics of the CAV can be expressed in the following second-order form
\begin{equation}\label{Eq:LinearCAVModel}
\begin{cases}
\dot{\tilde{s}}_0(t)=\tilde{v}_{-1}(t)-\tilde{v}_0(t),\\
\dot{\tilde{v}}_0(t)=u(t),\\
\end{cases}
\end{equation}
where the acceleration signal of the CAV can be designed according to explicit strategies, which also acts as the control input of the LCC system. Define the global state of LCC as
\begin{equation*}\label{Eq:LCCstate}
x(t)=\begin{bmatrix}\tilde{s}_{-m}(t),\tilde{v}_{-m}(t),\ldots,\tilde{s}_{0}(t),\tilde{v}_{0}(t),\ldots,\tilde{s}_{n}(t),\tilde{v}_{n}(t)\end{bmatrix}^T.
\end{equation*}
Then, based on the linearized HDVs' car-following model~\eqref{Eq:LinearHDVModel} and the CAV's dynamics~\eqref{Eq:LinearCAVModel}, the linearized state-space model for the LCC system is obtained
\begin{equation} \label{Eq:LCCSystemModel}
	\dot{x}(t)=Ax(t)+Bu(t)+H\tilde{v}_{\mathrm{h}}(t).
\end{equation}
The velocity error of the head vehicle $\tilde{v}_{\mathrm{h}}(t)$ can be regarded as an external disturbance signal into the system. The coefficient matrices $A\in \mathbb{R}^{(2n+2m+2)\times(2n+2m+2)}$, $B,H\in \mathbb{R}^{(2n+2m+2)\times 1}$ are given by
\begin{gather*}
A=\begin{bmatrix} P_1 & & & & &  \\
P_2 & P_1 & &  & &  \\
& \ddots& \ddots&  &  &\\
& & P_2& P_1 &  &\\
& & &  S_2 & S_1\\
& & & & P_2 &P_1\\
&&&&&\ddots&\ddots\\
&&&&&&P_2&P_1
\end{bmatrix},\\
B = \begin{bmatrix}
b_{-m}^T,\ldots,b_{-1}^T,b_0^T,b_1^T,\ldots,b_n^T
\end{bmatrix}^T,\\
H = \begin{bmatrix}
h_{-m}^T,\ldots,h_{-1}^T,h_0^T,h_1^T,\ldots,h_n^T
\end{bmatrix}^T,
\end{gather*}
where ($i\in \mathcal{F}\cup \mathcal{P}$, $j\in \{0\}\cup\mathcal{F}\cup \mathcal{P} \backslash \{-m\}$)
\begin{align*}
	&P_{1} = \begin{bmatrix} 0 & -1 \\ \alpha_{1} & -\alpha_{2} \end{bmatrix},
	P_{2} = \begin{bmatrix} 0 & 1 \\ 0 & \alpha_{3} \end{bmatrix};\,
	b_0 = \begin{bmatrix} 0  \\ 1 \end{bmatrix},
	b_i = \begin{bmatrix} 0 \\ 0 \end{bmatrix};\\
&	S_1 = \begin{bmatrix} 0 & -1 \\ 0 & 0 \end{bmatrix},
	S_2 = \begin{bmatrix} 0 & 1 \\ 0 & 0 \end{bmatrix};	\,
	h_{-m} = \begin{bmatrix} 1 \\ \alpha_{3}\end{bmatrix},
	h_j = \begin{bmatrix} 0 \\ 0 \end{bmatrix}.
\end{align*}

\begin{remark}
Limited perception abilities and imperfect reaction characteristics of human drivers have been recognized as crucial causes of traffic instabilities and even traffic jams~\cite{sugiyama2008traffic}. Wireless communications greatly enlarge the perception range for individual vehicles, but most current frameworks for CAVs in mixed traffic flow have been limited to focusing on the state of the HDVs ahead. One straightforward distinction of LCC lies in the explicit consideration of those vehicles behind. To highlight this distinction, we present two special cases of the general LCC system in the following.
\end{remark}

\begin{figure}[t]
	\vspace{1mm}
	\centering
	\includegraphics[scale=0.45]{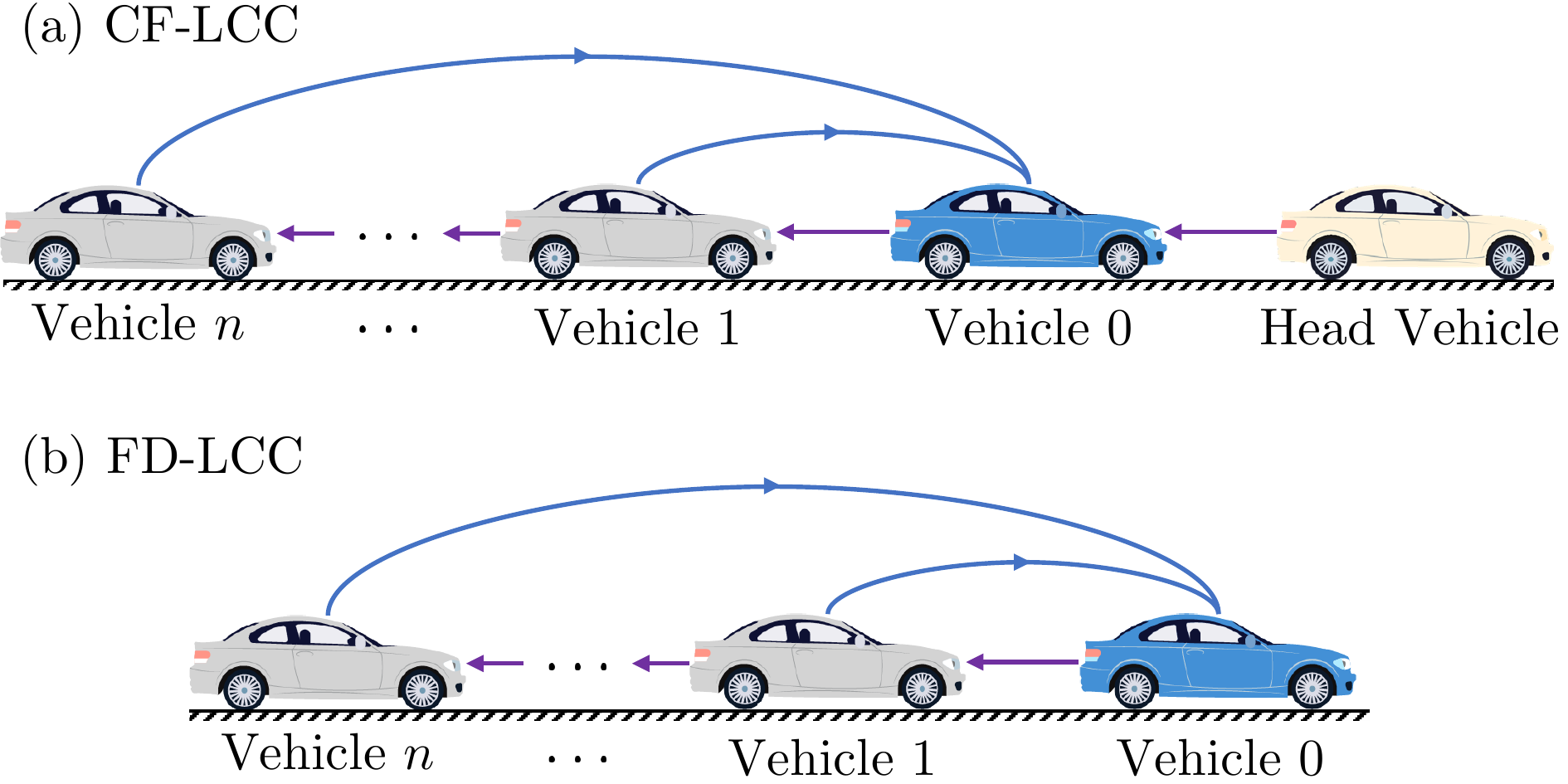}
	\vspace{-4mm}
	\caption{Two special cases of LCC when $\mathcal{P}=\emptyset$. (a) In CF-LCC, the CAV adopts the HDV's strategy \eqref{Eq:LinearHDVModel} to follow one single preceding vehicle, \ie, the head vehicle. (b) In FD-LCC, there are no preceding vehicles.}
	\label{Fig:FDLCC_CFLCC}
	\vspace{-5mm}
\end{figure}

\subsection{Two Special Cases of LCC}

It is known that the longitudinal behavior of individual vehicles in traffic flow includes two fundamental categories: car-following and free-driving \cite{treiber2013traffic}. Accordingly, we also present two special LCC systems, Car-Following LCC and Free-Driving LCC, as demonstrated in Fig.~\ref{Fig:FDLCC_CFLCC}. In both systems, we assume that the CAV receives no information from the preceding HDVs via V2V, \ie, $\mathcal{P}=\emptyset$, and focus on the LCC's ``looking behind'' characteristics.

\emph{Special Case 1:} Car-Following LCC System (CF-LCC)

In the first case, we assume that the CAV adopts the HDVs' control strategy \eqref{Eq:LinearHDVModel} to follow one single preceding vehicle, \ie, the head vehicle. Meanwhile, we also apply an additional control input $\hat{u} (t)$ into the CAV, which is determined by the state of the vehicles behind. The longitudinal dynamics of the CAV can thus be expressed by
\begin{equation}\label{Eq:CFLCC_CAVModel}
\begin{cases}
\dot{\tilde{s}}_0(t)=\tilde{v}_{\mathrm{h}}(t)-\tilde{v}_0(t),\\
\dot{\tilde{v}}_0(t)=\alpha_{1}\tilde{s}_0(t)-\alpha_{2}\tilde{v}_0(t)+\alpha_{3}\tilde{v}_{\mathrm{h}}(t)+\hat{u}(t).\\
\end{cases}
\end{equation}
In this case, the global state of the CF-LCC system degrades from that in \eqref{Eq:LCCSystemModel} to 
$\label{Eq:CFLCCstate}
x_{\mathrm{c}}(t)=[\tilde{s}_{0}(t),\tilde{v}_{0}(t),\tilde{s}_{1}(t),\tilde{v}_{1}(t),$ $\ldots,\tilde{s}_{n}(t),\tilde{v}_{n}(t)]^T.
$
Then the linearized state-space model for the CF-LCC system is presented as follows
\begin{equation} \label{Eq:CFLCCSystemModel}
\dot{x}_{\mathrm{c}}(t)=A_{\mathrm{c}}x_\mathrm{c}(t)+B_1 \hat{u}(t)+H_1 \tilde{v}_{\mathrm{h}}(t).
\end{equation}
The coefficient matrices $A_{\mathrm{c}}\in \mathbb{R}^{(2n+2)\times(2n+2)}$, $B_1,H_1\in \mathbb{R}^{(2n+2)\times 1}$ are given by
\begin{equation*}
	A_{\mathrm{c}}=\begin{bmatrix} P_1 & & & \\
	P_2 & P_1 & &    \\
	& \ddots& \ddots&  \\
	& & P_2& P_1 \\
	\end{bmatrix},\,
	B_1=\begin{bmatrix}
	0\\1\\0\\ \vdots \\0
	\end{bmatrix},\,
	H_1=\begin{bmatrix}
	1\\ \alpha_3 \\0\\ \vdots \\0
	\end{bmatrix}.
\end{equation*}

\emph{Special Case 2:} Free-Driving LCC System (FD-LCC)

The second case is to assume that the CAV is driving freely with no vehicles ahead. Consequently, the CAV's spacing $s_0$ has no real-world meaning and we consider its position instead as a state variable. Then, the longitudinal dynamics of the CAV can be given by a simple second-order integral model as follows
\begin{equation}
\begin{cases}
\dot{{p}}_0(t)=\tilde{v}_0(t),\\
\dot{\tilde{v}}_0(t)=u(t).\\
\end{cases}
\end{equation}
Defining the global state of the FD-LCC system as
$\label{Eq:FDLCCstate}
x_{\mathrm{f}}(t)=\begin{bmatrix}-{p}_{0}(t),\tilde{v}_{0}(t),\tilde{s}_{1}(t),\tilde{v}_{1}(t)\ldots,\tilde{s}_{n}(t),\tilde{v}_{n}(t)\end{bmatrix}^T,
$
where the negative sign exits for consistency with \eqref{Eq:LinearCAVModel} and \eqref{Eq:CFLCC_CAVModel}, then the linearized state-space model for the FD-LCC system is consequently obtained
\begin{equation} \label{Eq:FDLCCSystemModel}
\dot{x}_{\mathrm{f}}(t)=A_{\mathrm{f}}x_\mathrm{f}(t)+B_1 u(t),
\end{equation}
where $A_{\mathrm{f}}\in \mathbb{R}^{(2n+2)\times(2n+2)}$ is given by
\begin{equation*} \label{Eq:AfExpression}
A_{\mathrm{f}}=\begin{bmatrix} S_1 & & & \\
P_2 & P_1 & &    \\
& \ddots& \ddots&  \\
& & P_2& P_1 \\
\end{bmatrix}.
\end{equation*}

\section{Controllability of LCC Systems}
\label{Sec:3}

In this section, we analyze the controllability property of the LCC systems. According to the front-to-rear reaction dynamics of human drivers (see the purple arrows in Fig.~\ref{Fig:SystemSchematic}(c)), it is easy to understand that the CAV's behavior will have certain influence on its following HDVs, and our particular interest is to understand such influence from a control theoretic perspective.

As a fundamental property of dynamical systems, controllability is an essential metric to quantify the influence of the control input on the entire system. The formal definition is as follows.

\begin{definition}[Controllability \rm{\cite{skogestad2007multivariable}}]
	The dynamical system $\dot{x}=Ax+Bu$ is controllable, if for any initial state $x(0)=x_0$, any time $t_f>0$ and any final state $x_f$, there exists an input $u(t)$ such that $x(t_f)=x_f$.
\end{definition}

\begin{lemma}[PBH controllability test \rm{\cite{skogestad2007multivariable}}] \label{Lemma:PBH}
System $(A,B)$ is controllable, if and only if  $\left[  \lambda I-A,B \right]$  is of full row rank for all  $\lambda$  being an eigenvalue of  $ A$, where $I$ denotes an identity matrix with compatible dimension.
\end{lemma}

We first focus on the special case, the FD-LCC system, where there only exist several HDVs that follow behind the CAV. Our first result concerning the controllability of the LCC frameworks is as follows.

\begin{theorem} \label{Theorem:FDLCCControllability}
The FD-LCC system with no vehicle ahead and  $n$  HDVs behind given by \eqref{Eq:FDLCCSystemModel} is completely controllable, if the following condition \eqref{Eq:ControllabilityCondition} holds
\begin{equation}
\vspace{2mm}
\label{Eq:ControllabilityCondition}
	\alpha _{1}- \alpha _{2} \alpha _{3}+ \alpha _{3}^{2} \neq 0.
\end{equation}
\end{theorem}

\begin{proof}
Assume that the FD-LCC system \eqref{Eq:FDLCCSystemModel} is not completely controllable. According to Lemma \ref{Lemma:PBH}, this assumption indicates that there exists an eigenvalue  \(  \lambda  \)  of  \( A_{\mathrm{f}} \)  such that  \(  \left[  \lambda I-A_{\mathrm{f}},B_{1} \right]  \)  is not of full rank. Hence, there exists a nonzero vector  \(  \rho  \)  such that  \begin{subequations}
	\begin{gather}
		\rho ^{T} \left(  \lambda I-A_{\mathrm{f}} \right) =0; \label{Eq:RhoEquationA}\\
		\rho ^{T}B_{1}=0.  \label{Eq:RhoEquationB}
	\end{gather}
\end{subequations}
Denote  $\rho$  as
$
\rho = \begin{bmatrix}  \rho _{0}^{T}, \rho _{1}^{T}, \rho _{2}^{T}, \ldots , \rho _{n}^{T} \end{bmatrix}^{T},
$
where $\rho _{i}= \begin{bmatrix}  \rho _{i1}, \rho _{i2} \end{bmatrix}^{T} \in \mathbb{R}^{2 \times 1}$, $i=0,1, \ldots ,n $. Since only the second element in  $B_{1}$  is nonzero, \eqref{Eq:RhoEquationB} leads to  $\rho _{02}=0$. Substituting the expression of  $A_{\mathrm{f}}$ in \eqref{Eq:FDLCCSystemModel} into \eqref{Eq:RhoEquationA}, we have ($i \in \{1, \ldots ,n-1\}$)
\begin{subequations}
	\begin{align}
	 \rho _{0}^{T} \left( S_{1}- \lambda I \right) + \rho _{1}^{T}P_{2}&=0; \label{Eq:LittleRhoEquationA}\\
	\rho _{i}^{T} \left( P_{1}- \lambda I \right) + \rho _{i+1}^{T}P_{2}&=0;
	\label{Eq:LittleRhoEquationB}\\
\rho _{n}^{T} \left( P_{1}- \lambda I \right) &=0. \label{Eq:LittleRhoEquationC}
\end{align}
\end{subequations}

\emph{Case 1:} $\lambda ^{2}+ \alpha _{2} \lambda + \alpha _{1} \neq 0 $. Then,  \( P_{1}- \lambda I\)  is nonsingular. According to \eqref{Eq:LittleRhoEquationC}, we have  \(  \rho _{n}^{T}=0 \). Substituting  \(  \rho _{n}^{T}=0 \)  into \eqref{Eq:LittleRhoEquationB}, we have  \(  \rho _{n-1}^{T}=0 \). Using \eqref{Eq:LittleRhoEquationB} recursively, we can obtain that  \(  \rho _{i}^{T}=0,\,i=1, \ldots ,n \), which also leads to  \(  \rho _{0}^{T} \left( S_{1}- \lambda I \right) =0 \). Expanding this equation, we have  \(  \rho _{01}+ \lambda  \rho _{02}=0 \). Since  \(  \rho _{02}=0\), it is obtained that  \(  \rho _{01}=0 \). Consequently, we arrive at  \(  \rho =0 \), which contradicts the condition that  \(  \rho  \)  is a nonzero vector.

\emph{Case 2:} $\lambda ^{2}+ \alpha _{2} \lambda + \alpha _{1} = 0 $. In this case, we have  \(  \lambda  \neq 0 \)  since  \(  \alpha _{1}>0 \). Also, it can be obtained that  \(  \alpha _{3} \lambda + \alpha _{1} \neq 0 \); otherwise, condition \eqref{Eq:ControllabilityCondition} will be contradicted. Expanding \eqref{Eq:LittleRhoEquationA} leads to  \(  \lambda  \rho _{01}=0 \)  and  \(  \rho _{01}+ \lambda  \rho _{02}+ \rho _{11}+ \alpha _{3} \rho _{12}=0 \). Hence, we have  \(  \rho _{01}=0 \)  and  \(  \rho _{11}+ \alpha _{3} \rho _{12}=0 \). Meanwhile, letting  \( i=1 \)  and expanding \eqref{Eq:LittleRhoEquationB} yields  \(  \rho _{12}=\frac{ \lambda }{ \alpha _{1}} \rho _{11} \), which, combined with  \(  \rho _{11}+ \alpha _{3} \rho _{12}=0 \), leads to  \(  \rho _{11}= \rho _{12}=0 \). Letting  \( i=2, \ldots ,n \)  and expanding \eqref{Eq:LittleRhoEquationB} and \eqref{Eq:LittleRhoEquationC}, we can obtain the following results (\( i=2, \ldots ,n \))
\begin{subequations}
\begin{align}
\lambda  \rho _{i1}- \alpha _{1} \rho _{i2}&=0, \label{Eq:ControllabilityCase2EquationA}\\
\left(  \lambda ^{2}+ \alpha _{2} \lambda + \alpha _{1} \right)  \rho _{ \left( i-1 \right) 1}&=\left(  \alpha _{3} \lambda + \alpha _{1} \right)  \rho _{i1} . \label{Eq:ControllabilityCase2EquationB}
\end{align}
\end{subequations}
Since  \(  \alpha _{3} \lambda + \alpha _{1} \neq 0\), substituting  \(  \lambda ^{2}+ \alpha _{2} \lambda + \alpha _{1}=0 \) into~\eqref{Eq:ControllabilityCase2EquationB} and then \eqref{Eq:ControllabilityCase2EquationA} yields  \(  \rho _{i1}= \rho _{i2}=0,\,i=2, \ldots ,n \). Accordingly, we arrive at  \(  \rho =0 \), which contradicts  \(  \rho  \neq 0\).

In summary, the assumption does not hold, and we can conclude that the FD-LCC system \eqref{Eq:FDLCCSystemModel} is controllable. This completes the proof of Theorem \ref{Theorem:FDLCCControllability}.
\end{proof}

Note that \eqref{Eq:ControllabilityCondition} is a sufficient condition. At the random choice of  \(  \alpha _{1}, \alpha _{2}, \alpha _{3} \), condition \eqref{Eq:ControllabilityCondition} is satisfied with probability one; accordingly, FD-LCC is controllable with probability one. 
Regarding the controllability of the CF-LCC system \eqref{Eq:CFLCCSystemModel} and the general LCC system \eqref{Eq:LCCSystemModel}, we have the following two results. Due to page limit, a detailed proof will be presented in an extended version.

\begin{corollary} \label{Corollary:CFLCCControllability}
	The CF-LCC system where the CAV adopts the HDVs' dynamics \eqref{Eq:LinearHDVModel} to follow one vehicle ahead and considers $n$ HDVs behind given by \eqref{Eq:CFLCCSystemModel} is completely controllable, if the condition \eqref{Eq:ControllabilityCondition} holds.
\end{corollary}

\begin{theorem} \label{Theorem:LCCControllability}
Consider the general LCC system with $m$ vehicles ahead and $n$ HDVs behind given by \eqref{Eq:LCCSystemModel}. The following statements hold:
\begin{enumerate}
	\vspace{-1.5mm}
	\item The subsystem consisting of the states of the vehicles ahead, \ie,  \( \tilde{s}_{-m} \left( t \right) ,\tilde{v}_{-m} \left( t \right) , \ldots ,\tilde{s}_{-1} \left( t \right) ,\tilde{v}_{-1} \left( t \right)  \)  is uncontrollable.
	\item The subsystem consisting of the states of the CAV and the vehicles behind, \ie,  $ \tilde{s}_{0} \left( t \right) ,\tilde{v}_{0} \left( t \right) ,$   $\tilde{s}_{1} \left( t \right) ,\tilde{v}_{1} \left( t \right) , \ldots ,\tilde{s}_{n} \left( t \right) ,\tilde{v}_{n} \left( t \right)$, is controllable, if the condition \eqref{Eq:ControllabilityCondition} holds.
\end{enumerate}
\end{theorem}

\begin{remark}
The physical interpretation of Theorem \ref{Theorem:LCCControllability} is that the control input of the CAV has no influence on the preceding HDVs, but has complete control of the motion of the following HDVs. In CCC-type frameworks, the only controllable part is the state of the CAV itself, \ie, $\tilde{s} _0 (t),\tilde{v} _0 (t)$~\cite{jin2017optimal}; consequently, the control objective of CCC is limited to improving the performance of the CAV's own car-following behavior. By contrast, the controllability of the state of the vehicles behind allows the CAV to act as \emph{a sophisticated leader with global consideration}, \eg, aiming to improve the performance of the entire upstream traffic flow. Note that this result also generalizes the stabilizability results in the closed ring-road traffic system \cite{zheng2020smoothing,wang2020controllability}, where it has been shown that one single CAV can stabilize the entire traffic flow.
\end{remark}

\section{Head-to-Tail String Stability}
\label{Sec:4}

The controllability analysis reveals the potential of the CAV to actively lead the motion of the following HDVs. As for the vehicles ahead, the CAV's capability in dampening front perturbations has received considerable attention in existing research, where string stability is a significant notion to describe this capability. Along this direction, we study the string stability performance of the mixed traffic flow under the proposed LCC framework in this section.

\subsection{Head-to-Tail Transfer Function}

In the local control of an individual vehicle, string stability depicts its ability in attenuating velocity fluctuations coming from the vehicle immediately ahead. For a series of vehicles, head-to-tail string stability is utilized more often, which is defined as follows.

\begin{definition}[Head-to-Tail String Stability \rm{\cite{jin2014dynamics}}]
	 \label{Def:HeadtoTail}
Given a series of consecutive vehicles, denote the velocity deviation of the vehicle at the head and the one at the tail as  $ \tilde{v}_\mathrm{h} \left( t \right) $  and  \( \tilde{v}_\mathrm{t} \left( t \right)  \) , respectively. The head-to-tail transfer function is defined as
\begin{equation} \label{Eq:GeneralTransferFunction}
	\Gamma (s) = \frac{\widetilde{V}_\mathrm{t} (s)  }{\widetilde{V}_\mathrm{h} (s) },
\end{equation}
where $\widetilde{V}_\mathrm{h}(s), \widetilde{V}_\mathrm{t}(s)$ denote the Laplace transform of  $ \tilde{v}_\mathrm{h} (t) $  and  $ \tilde{v}_\mathrm{t} (t) $, respectively. Then head-to-tail string stability holds if and only if
\begin{equation} \label{Eq:HeadtoTailDefinition}
	\vert  \Gamma  \left( j \omega  \right)  \vert ^{2}<1,  \forall  \omega >0,
\end{equation}
where  $ j=\sqrt[]{-1} $, and $\vert \cdot \vert  $ denotes the modulus.
\end{definition}

Definition \ref{Def:HeadtoTail} shows that head-to-tail string stability describes a property in a series of vehicles where the perturbation signals are attenuated between the head and the tail vehicles for all excitation frequencies. When head-to-tail string stability is violated, a small perturbation in the head vehicle might cause severe stop-and-go behaviors in the following vehicles, causing great loss of travel efficiency and high risk of traffic accidents.

We then proceed to investigate the head-to-tail string stability property of LCC. As shown in Fig.~\ref{Fig:SystemSchematic}(c), we consider a general scenario with $m$ preceding HDVs and $n$ following HDVs. The velocity perturbation $\tilde{v}_\mathrm{h} (t)$ of the head vehicle and the velocity perturbation $\tilde{v}_{n} (t)$ of the HDV at the very tail are regarded as the input and the output, respectively.

Based on the Laplace transform of the linearized car-following model \eqref{Eq:LinearHDVModel} of HDVs, the local transfer function of HDVs' dynamics is obtained as follows ($i \in \mathcal{F} \cup \mathcal{P}$)
\begin{equation}\label{Eq:HDVTransferFunction}
\frac{\widetilde{V}_{i} \left( s \right) }{\widetilde{V}_{i-1} \left( s \right) }=\frac{ \alpha _{3}s+ \alpha _{1}}{s^{2}+ \alpha _{2}s+ \alpha _{1}}=\frac{ \varphi  \left( s \right) }{ \gamma  \left( s \right) },
\end{equation}
with
$
\varphi  \left( s \right) = \alpha _{3}s+ \alpha _{1},  \gamma  \left( s \right) =s^{2}+ \alpha _{2}s+ \alpha _{1}.
$

As for the CAV, we assume that it adopts the HDVs' strategy \eqref{Eq:LinearHDVModel} to follow the vehicle immediately ahead, while also exploiting the state of surrounding HDVs for feedback control. Denote  $  \mu _{i},k_{i} $  as the feedback gain corresponding to the spacing error and the velocity error of vehicle  $ i $  ($ i \in \mathcal{F} \cup \mathcal{P} $), respectively. Then,  the control input can be expressed as
\begin{equation}\label{Eq:ControlInputDefinition}
u \left( t \right) = \alpha _{1}\tilde{s}_{0}- \alpha _{2}\tilde{v}_{0}+ \alpha _{3}\tilde{v}_{-1}
+  \sum _{i\in \mathcal{F} \cup \mathcal{P}} \left(  \mu _{i}\tilde{s}_{i} \left( t \right) +k_{i}\tilde{v}_{i} \left( t \right)  \right).
\end{equation}
Substituting \eqref{Eq:ControlInputDefinition} into the CAV's longitudinal dynamics \eqref{Eq:LinearCAVModel} and combining the Laplace transform of \eqref{Eq:LinearHDVModel} and \eqref{Eq:LinearCAVModel}, we can obtain the head-to-tail transfer function of the LCC system as follows
\begin{equation}\label{Eq:LCCTransferFunction}
\Gamma  \left( s \right) =\frac{ \varphi  \left( s \right) + \sum _{i\in \mathcal{P}}H_{i} \left( s \right)  ( \frac{ \varphi  \left( s \right) }{ \gamma  \left( s \right) } ) ^{i+1}}{ \gamma  \left( s \right) - \sum_{i \in \mathcal{F}} H_{i} \left( s \right)  ( \frac{ \varphi  \left( s \right) }{ \gamma  \left( s \right) } ) ^{i}} \cdot \left( \frac{ \varphi  \left( s \right) }{ \gamma  \left( s \right) } \right) ^{n+m},
\end{equation}
where
$
	H_{i} \left( s \right) = \mu _{i}\left(\gamma (  s ) / \varphi  ( s ) - 1\right)+k_{i}s$, $i \in \mathcal{F} \cup \mathcal{P}
$.

We consider the following special cases to facilitate the understanding of the transfer function \eqref{Eq:LCCTransferFunction}.

	1) When  $\mu _{i}=k_{i}=0$, $i \in\mathcal{F} \cup \mathcal{P}$ , the CAV follows the same control strategy as that of the human drivers. The head-to-tail transfer function then degrades to
	\begin{equation}
	\Gamma _{1} \left( s \right) = \left( \frac{ \varphi  \left( s \right) }{ \gamma  \left( s \right) } \right) ^{n+m+1},
	\end{equation}
	corresponding to a platoon of  $n+m+1$  HDVs.
	
	2) When $\mu _{i}=k_{i}=0$, $i \in \mathcal{F} $, \ie, the CAV only exploits the information of multiple vehicles ahead for longitudinal control, corresponding to a typical CCC-type strategy, the head-to-tail transfer function then becomes
	\begin{equation} \label{Eq:AheadTransferFunction}
	\Gamma_2  \left( s \right) =\frac{ \varphi  \left( s \right) + \sum _{i\in \mathcal{P}}H_{i} \left( s \right)  ( \frac{ \varphi  \left( s \right) }{ \gamma  \left( s \right) } ) ^{i+1}}{ \gamma  \left( s \right)} \cdot \left( \frac{ \varphi  \left( s \right) }{ \gamma  \left( s \right) } \right) ^{n+m}.
	\end{equation}

3) When  $ \mu _{i}=k_{i}=0$, $i \in \mathcal{P}$, the CAV adopts a same strategy as that of HDVs to follow the head vehicle, and also monitors multiple HDVs behind to adjust its own motion. In this case, the head-to-tail transfer function becomes
\begin{equation} \label{Eq:BehindTransferFunction}
\Gamma_3  \left( s \right) =\frac{ \varphi  \left( s \right)  }{ \gamma  \left( s \right)- \sum_{i \in \mathcal{F}} H_{i} \left( s \right)  ( \frac{ \varphi  \left( s \right) }{ \gamma  \left( s \right) } ) ^{i}} \cdot \left( \frac{ \varphi  \left( s \right) }{ \gamma  \left( s \right) } \right) ^{n+m}.
\end{equation}

\begin{figure}[t]
	\vspace{1mm}
	\centering
	\subfigure[]
	{\includegraphics[scale=0.5,trim=5 2 17 10,clip]{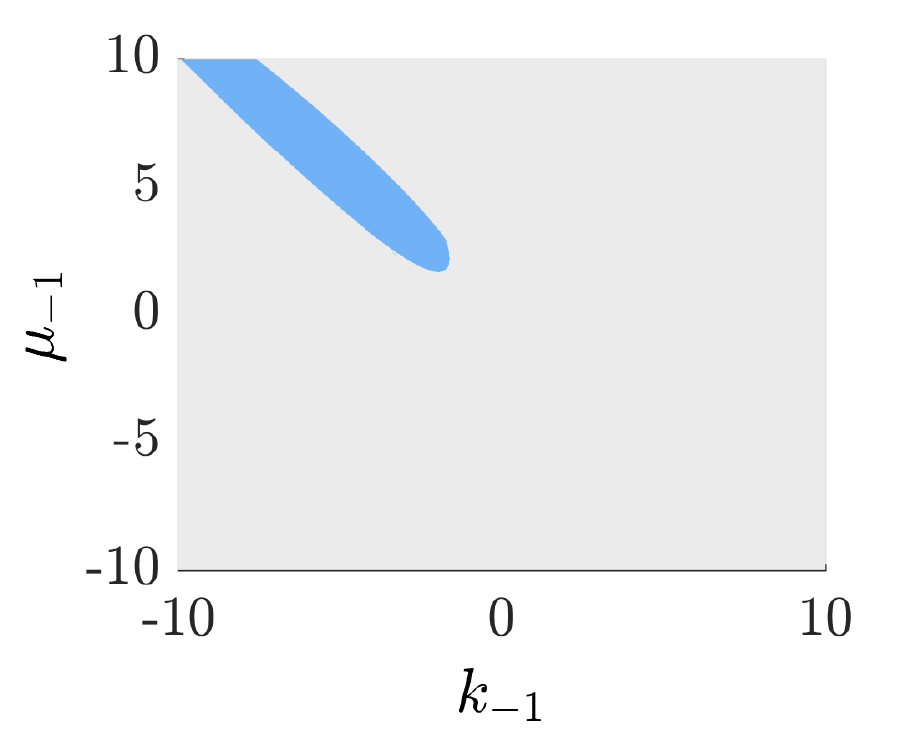}}
	\subfigure[]
	{\includegraphics[scale=0.5,trim=5 2 17 10,clip]{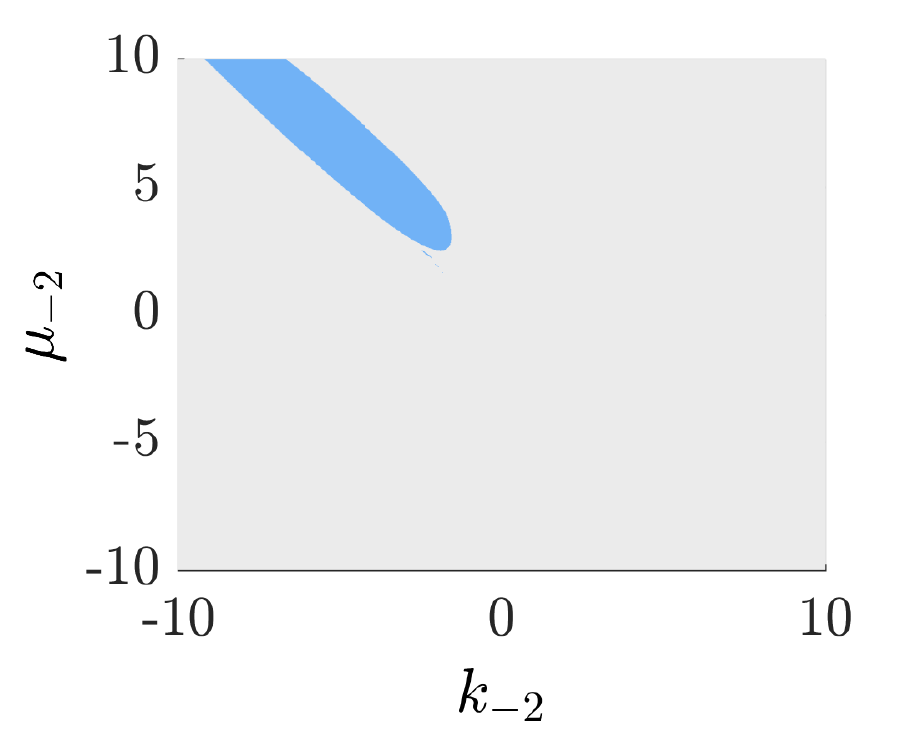}}\\
	\subfigure[]
	{\includegraphics[scale=0.5,trim=5 2 17 10,clip]{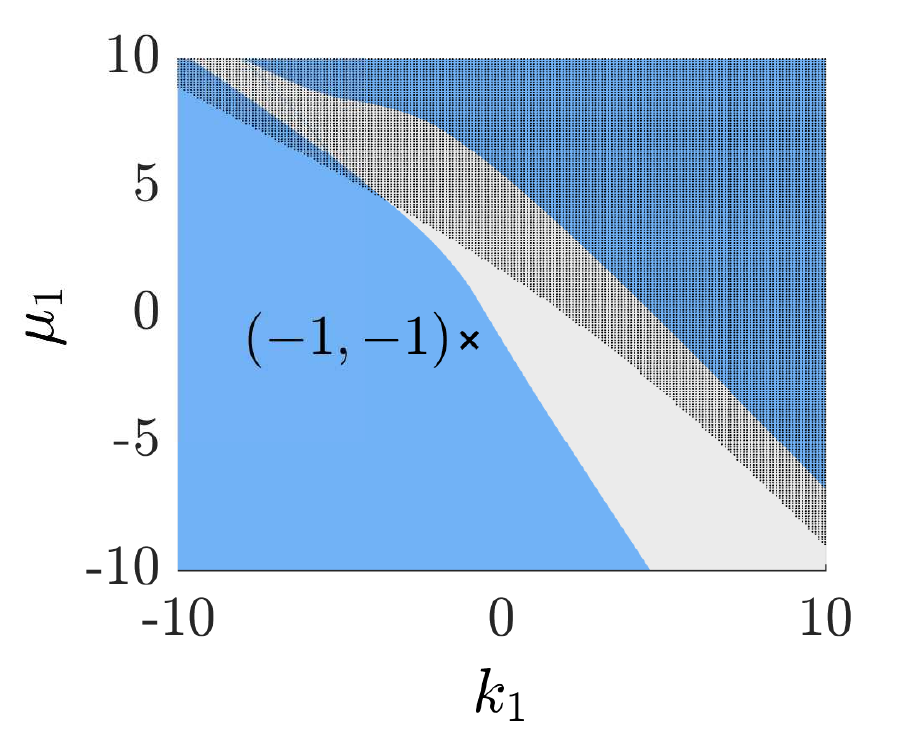}}
	\subfigure[]
	{\includegraphics[scale=0.5,trim=5 2 17 10,clip]{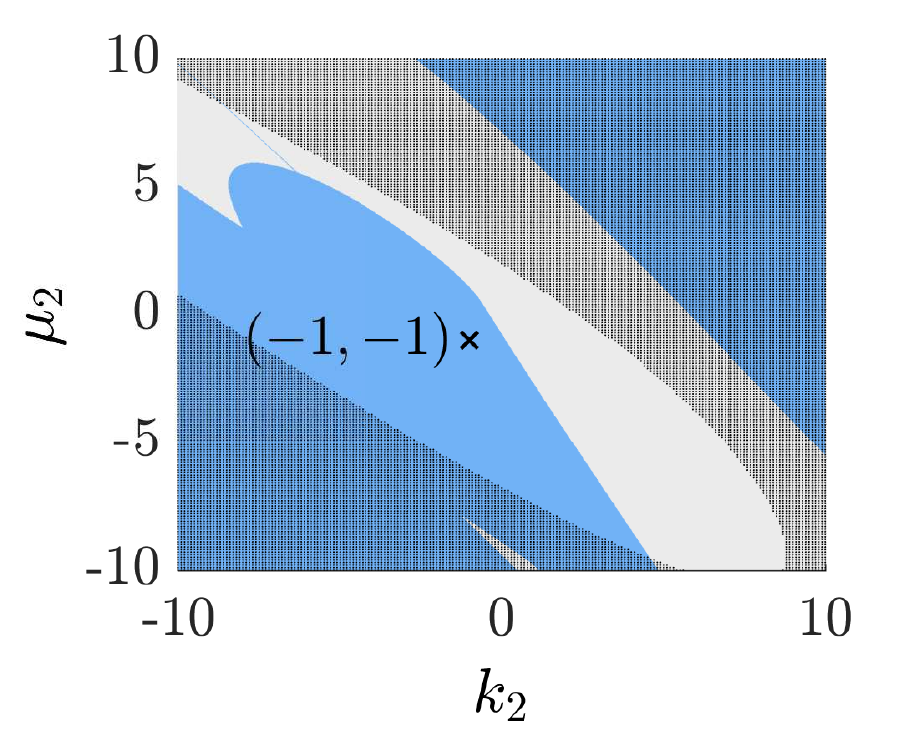}}
	\vspace{-5mm}
	\caption{Head-to-tail string stability charts ($n=m=2$) when the CAV monitors one single vehicle. (a)(b) corresponds to the vehicles ahead, while (c)(d) corresponds to the vehicles behind.}
	\label{Fig:ExperimentA}
	\vspace{-5mm}
\end{figure}

\begin{remark}
	As shown in \eqref{Eq:AheadTransferFunction} and \eqref{Eq:BehindTransferFunction}, the incorporation of either the state of the preceding vehicles or the following vehicles will bring a significant change to the head-to-tail transfer characteristic of mixed traffic flow, but the changes of the two types work in different ways. Much existing research has addressed the head-to-tail string stability of mixed traffic flow when the CAV monitors the motion of the vehicles ahead; see, \eg, \cite{di2019cooperative,jin2014dynamics,zhou2020stabilizing}. However, the influence of the incorporation of the vehicles behind on string stability remains unclear. Moreover, previous research mostly focuses on the transfer function from the head vehicle to the CAV itself, instead of a certain vehicle behind the CAV. It is worth noting that \emph{the perturbations will continue to propagate upstream after reaching the CAV}, which means that although the CAV is exploited to mitigate the perturbations coming from front, they might still be amplified behind the CAV. This determines the necessity of incorporating the motion of the vehicles behind into the CAV's longitudinal control.
\end{remark}

\subsection{Region of Head-to-Tail String Stability}

\begin{figure}[t]
	\vspace{1mm}
	\centering
	\subfigure[]
	{\includegraphics[scale=0.5,trim=5 2 17 10,clip]{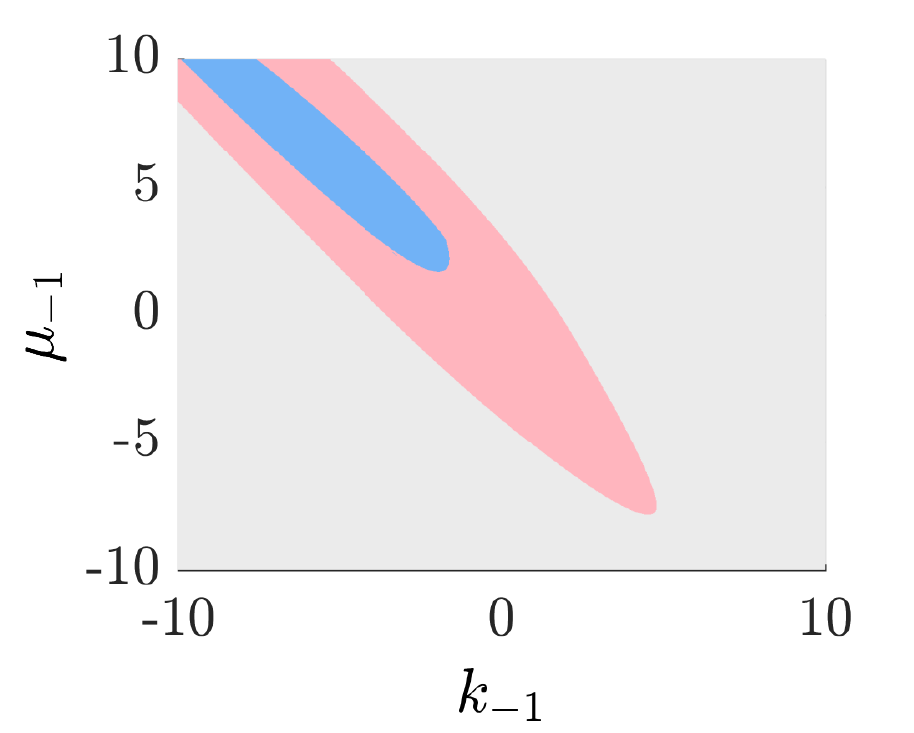}}
	\subfigure[]
	{\includegraphics[scale=0.5,trim=5 2 17 10,clip]{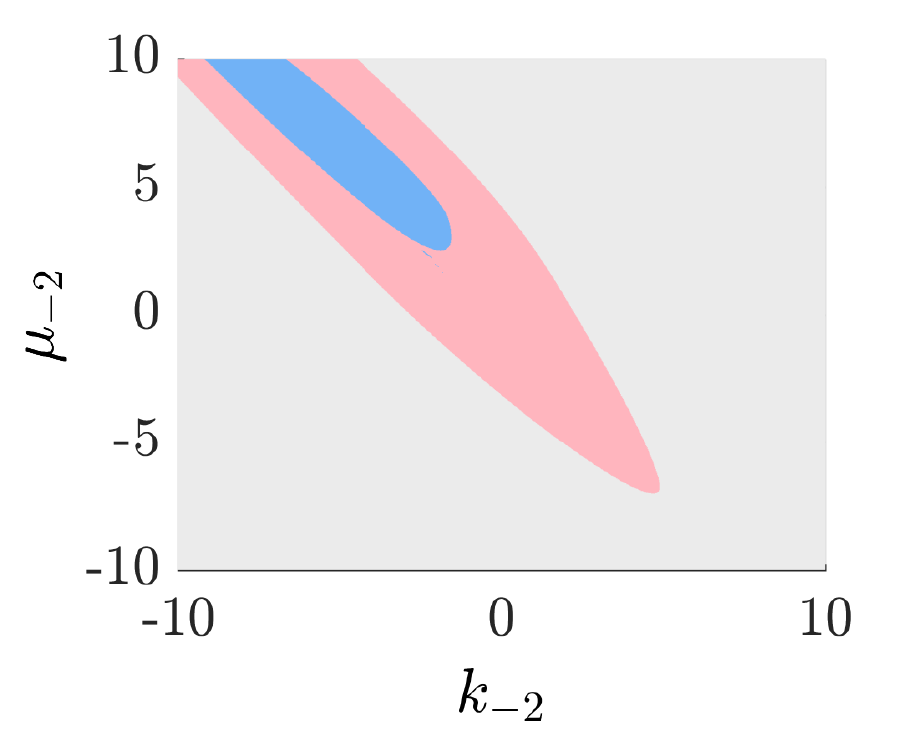}}\\
	\subfigure[]
	{\includegraphics[scale=0.5,trim=5 2 17 10,clip]{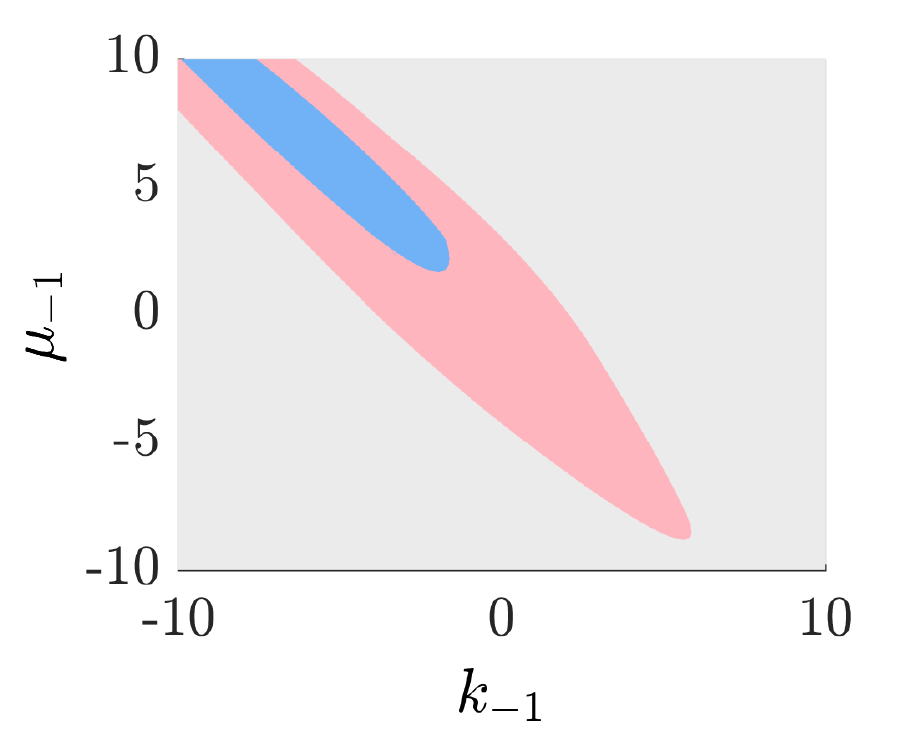}}
	\subfigure[]
	{\includegraphics[scale=0.5,trim=5 2 17 10,clip]{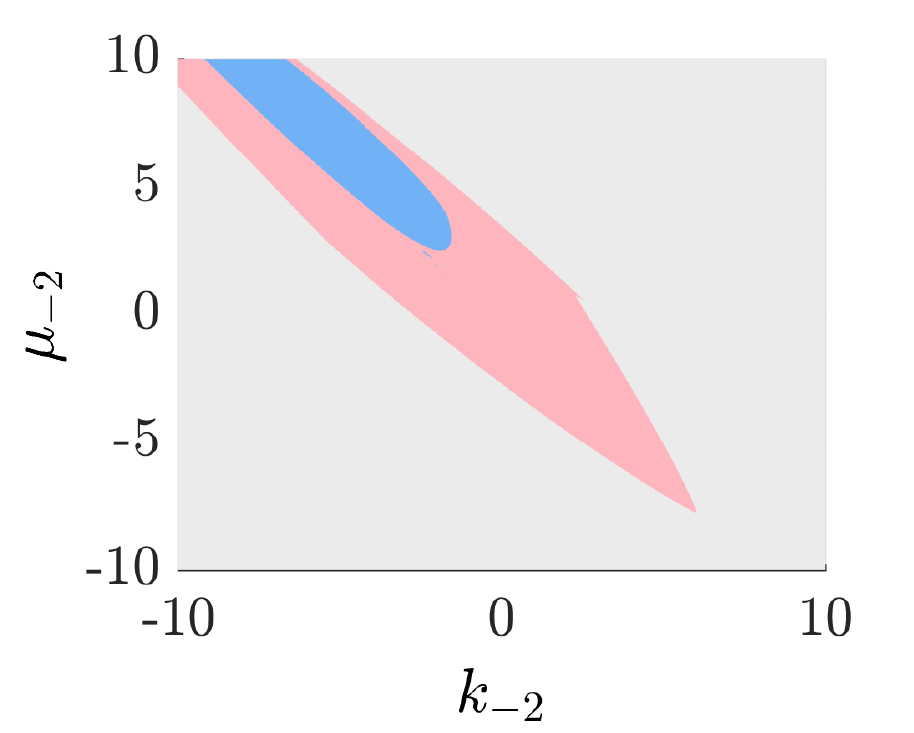}}
	\vspace{-4mm}
	\caption{Head-to-tail string stability charts ($n=m=2$) after incorporation of one vehicle behind. (a)(b) $\mu_1=k_1=-1$. (c)(d) $\mu_2=k_2=-1$.}
	\label{Fig:ExperimentB}
	\vspace{-5mm}
\end{figure}

We proceed to numerically solve the head-to-tail string stable regions of \eqref{Eq:LCCTransferFunction}. The OVM model \eqref{Eq:OVMmodel} is employed, where  $\alpha =0.6, \beta =0.9, v_{\max }=30,s_{\mathrm{st}}=5,s_{\mathrm{go}}=35$. This is a typical setup of HDVs, yielding a local string unstable car-following behavior \cite{jin2017optimal}.

The first numerical study aims to analyze the feedback gains separately in \eqref{Eq:ControlInputDefinition} based on Definition \ref{Def:HeadtoTail}. Precisely, we assume that each time the CAV only monitors one single HDV for feedback control, \ie, the feedback gains of other vehicles are fixed to zeros. The head-to-tail string stable results are demonstrated in the blue domains in Fig.~\ref{Fig:ExperimentA}, and we also shade the areas where plant stability is not guaranteed. As can be clearly observed, the ``looking behind"  string stable areas are apparently larger than the ``looking ahead" ones, representing more string stable choices, despite a constraint on plant stable regions.

The second numerical study aims to investigate the influence of ``looking behind"  on the original string stable areas of ``looking ahead"  feedback gains in Fig.~\ref{Fig:ExperimentA}(a)(b). We consider two fixed cases (marked in Fig.~\ref{Fig:ExperimentA}(c)(d)), and the new results of ``looking ahead"  string stable regions under  $\mu _{1}=-1,k_{1}=-1$  or  $\mu _{2}=-1,k_{2}=-1$  are illustrated in Fig.~\ref{Fig:ExperimentB}. The blue areas denote the string stable regions when monitoring one single preceding vehicle, which remain the same as those in Fig.~\ref{Fig:ExperimentA}(a)(b). Interestingly, after the CAV takes into additional consideration the motion of one following vehicle, either vehicle 1 or vehicle 2, the ``looking ahead" string stable regions witness a significant expansion, which is highlighted in red domains in Fig.~\ref{Fig:ExperimentB}, indicating more feasible ``looking ahead'' feedback gains to ensure head-to-tail string stability.

\subsection{Case Study in Mitigating Traffic Perturbations}

The aforementioned two numerical studies reveal that the incorporation of the following vehicles contributes to more string stable options for CAV's feedback policies. We then choose several string stable cases (shown in Table~\ref{Tb:FeedbackGainSetup}) to show the quantificational improvement of the CAV's capability in mitigating perturbations after ``looking behind" . The magnitude of the transfer function \eqref{Eq:LCCTransferFunction} is illustrated in Fig.~\ref{Fig:ExperimentC}(a). In Case A, when the CAV only monitors one preceding vehicle, the magnitude profile is beneath one at all frequencies, indicating string stable performance, compared to the HDV-only scenario. The magnitude witnesses an apparent drop after extra consideration of vehicle 1 (Case B), and continues to decrease after additional incorporation of vehicle 2 (Case C). From the time-domain perspective, Fig.~\ref{Fig:ExperimentC}(b)-(d) shows that the amplitude of velocity fluctuations of the following vehicles becomes smaller from Case A to Case C under a same slight perturbation of the head vehicle. Accordingly, both the frequency-domain and the time-domain observations indicate that the degree to which the CAV mitigates the perturbation becomes higher after ``looking behind"  appropriately. This result reveals another potential of LCC, which could further improve the CAV's capability in dampening traffic waves and smoothing traffic flow, compared with traditional ``looking ahead" only strategies.

\begin{figure}[t]
	\vspace{1mm}
	\centering
	\subfigure[]
	{\includegraphics[scale=0.36,height=3cm]{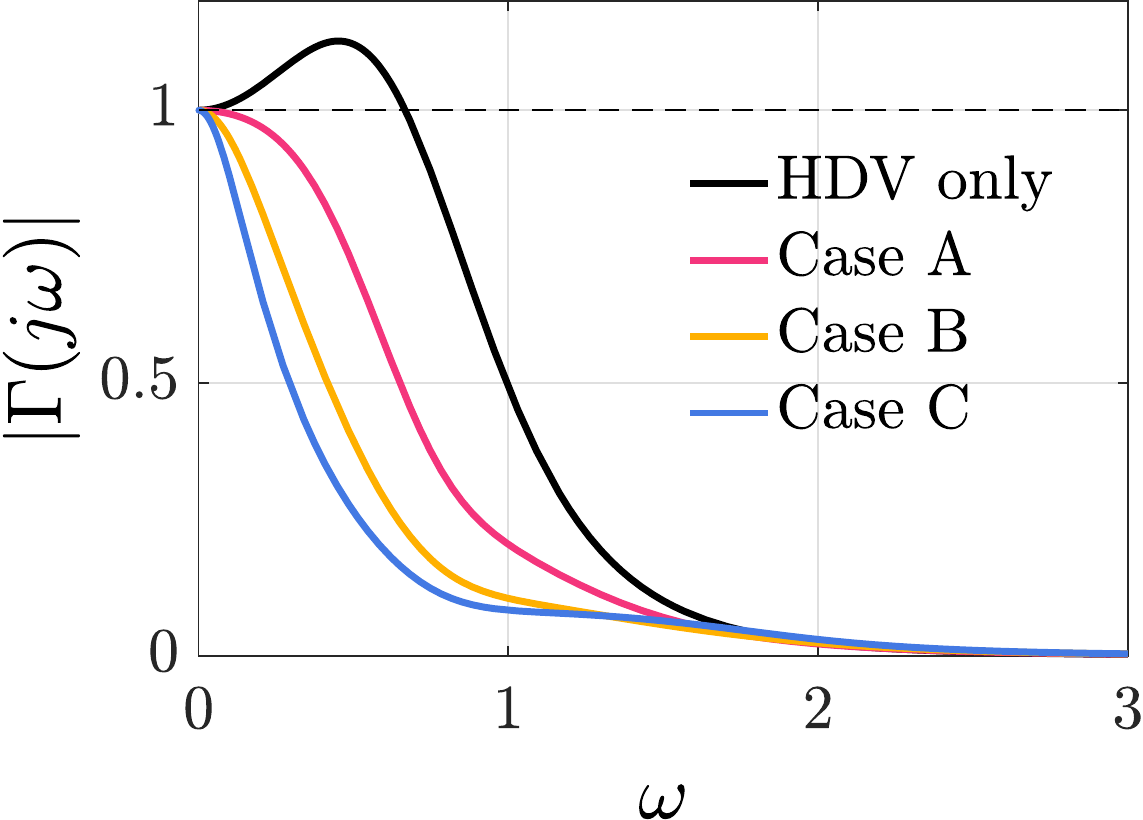}}
	\subfigure[Case A]
	{\includegraphics[scale=0.36,height=3cm]{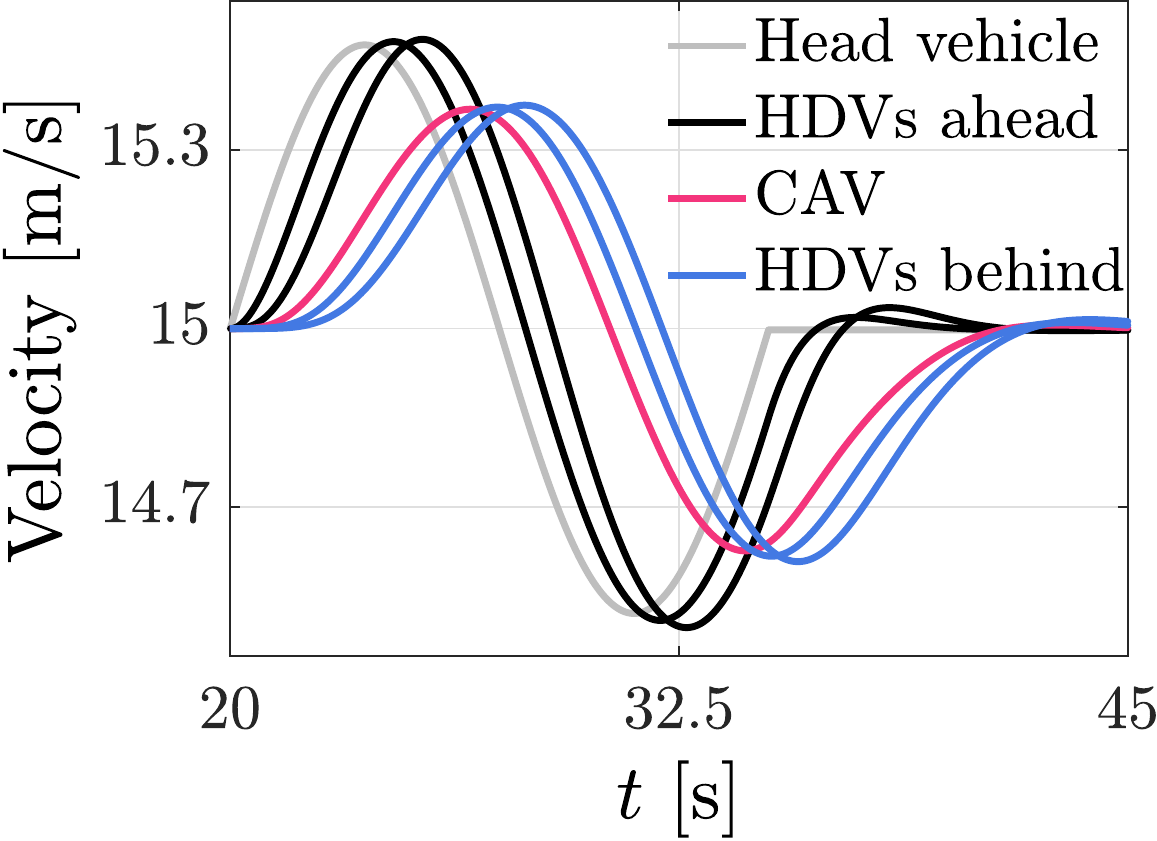}}
	\subfigure[Case B]
	{\includegraphics[scale=0.36,height=3cm]{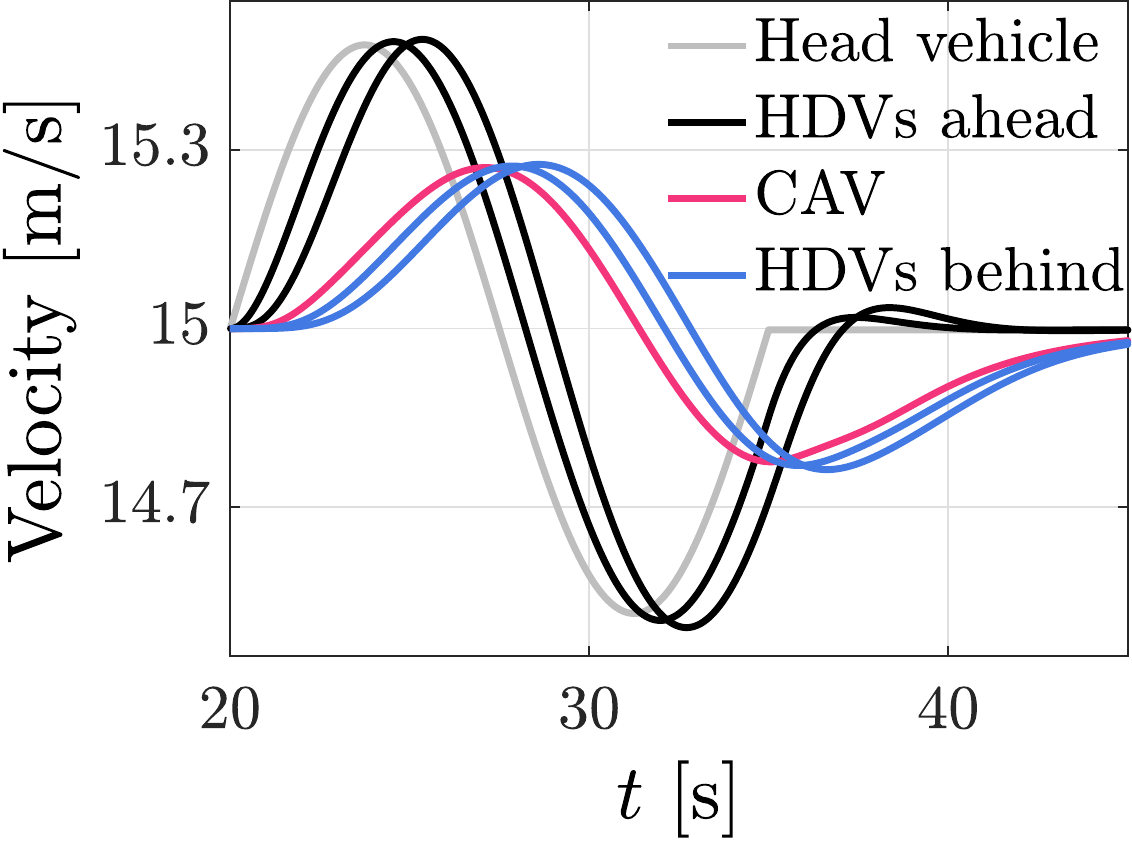}}
	\subfigure[Case C]
	{\includegraphics[scale=0.36,height=3cm]{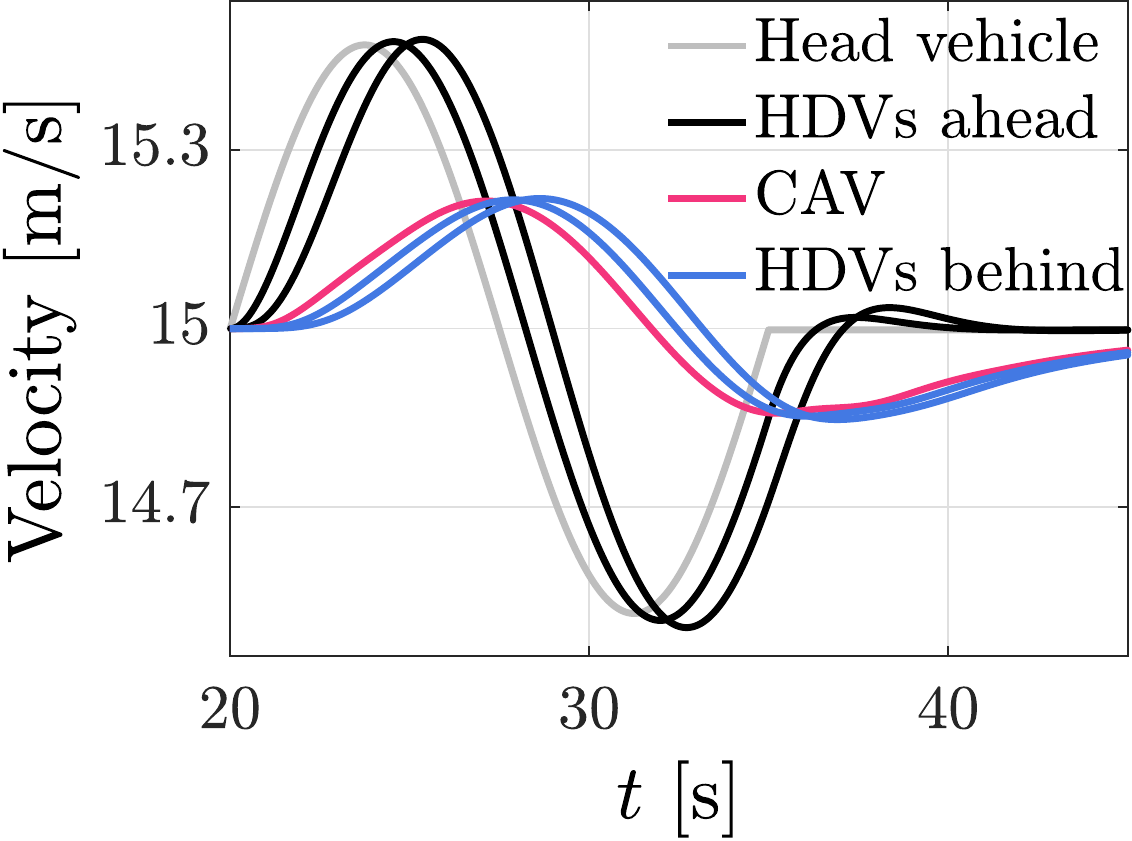}}
	\vspace{-2mm}
	\caption{Response profile when $n=m=2,\mu_{-2}=k_{-2}=0$. (a) Frequency-domain response: magnitude of transfer function \eqref{Eq:LCCTransferFunction}. (b)(c)(d) Time-domain response: velocity profile of each vehicle. }
	\label{Fig:ExperimentC}
\end{figure}

\begin{table}[t]
	\vspace{-2mm}
	\begin{center}
		\caption{Parameter Setup in Feedback Gains}\label{Tb:FeedbackGainSetup}
		\begin{tabular}{ccccccc}
			& $\mu_{-1}$ & $k_{-1}$ & $\mu_{1}$ & $k_{1}$ & $\mu_{2}$ & $k_{2}$  \\\hline
			Case A & 3 & -3 & 0  & 0 & 0 & 0\\
			Case B & 3 & -3 & -1  & -1 & 0 & 0\\
			Case C & 3 & -3 & -1  & -1 & -1 & -1\\
			\hline
		\end{tabular}
	\end{center}
	\vspace{-7mm}
\end{table}

\vspace{-0.5mm}
\section{Conclusions}
\vspace{-0.5mm}
\label{Sec:5}
In this paper, we have introduced the notion of Leading Cruise Control (LCC). Based on the dynamical model of the general LCC system and two special cases, we prove the controllability of the state of the vehicles behind through CAV's active action, and investigate the head-to-tail string stability of the proposed framework. These results reveal a great potential of incorporating vehicles behind into CAV's control, which could take full advantage of V2V connectivity and vehicular automation, contributing to a further improvement to the mixed traffic flow where HDVs and CAVs coexist in the near future. Future directions include designing more sophisticated controllers for LCC and analyzing the influence of behavior uncertainty and reaction time of human drivers on the LCC performance.

\addtolength{\textheight}{-12cm}   








\bibliographystyle{IEEEtran}
\bibliography{IEEEabrv,mybibfile_short}

\begin{thebibliography}{10}
\providecommand{\url}[1]{#1}
\csname url@samestyle\endcsname
\providecommand{\newblock}{\relax}
\providecommand{\bibinfo}[2]{#2}
\providecommand{\BIBentrySTDinterwordspacing}{\spaceskip=0pt\relax}
\providecommand{\BIBentryALTinterwordstretchfactor}{4}
\providecommand{\BIBentryALTinterwordspacing}{\spaceskip=\fontdimen2\font plus
\BIBentryALTinterwordstretchfactor\fontdimen3\font minus
  \fontdimen4\font\relax}
\providecommand{\BIBforeignlanguage}[2]{{%
\expandafter\ifx\csname l@#1\endcsname\relax
\typeout{** WARNING: IEEEtran.bst: No hyphenation pattern has been}%
\typeout{** loaded for the language `#1'. Using the pattern for}%
\typeout{** the default language instead.}%
\else
\language=\csname l@#1\endcsname
\fi
#2}}
\providecommand{\BIBdecl}{\relax}
\BIBdecl

\bibitem{contreras2017internet}
J.~Contreras-Castillo, S.~Zeadally \emph{et~al.}, ``Internet of vehicles:
  architecture, protocols, and security,'' \emph{IEEE Internet Things J.},
  vol.~5, no.~5, pp. 3701--3709, 2017.

\bibitem{vahidi2003research}
A.~Vahidi and A.~Eskandarian, ``Research advances in intelligent collision
  avoidance and adaptive cruise control,'' \emph{IEEE Trans. Intell. Transp.
  Syst.}, vol.~4, no.~3, pp. 143--153, 2003.

\bibitem{li2017dynamical}
S.~E. Li, Y.~Zheng \emph{et~al.}, ``Dynamical modeling and distributed control
  of connected and automated vehicles: Challenges and opportunities,''
  \emph{IEEE Intell. Transp. Syst. Mag.}, vol.~9, no.~3, pp. 46--58, 2017.

\bibitem{orosz2016connected}
G.~Orosz, ``Connected cruise control: modelling, delay effects, and nonlinear
  behaviour,'' \emph{Veh. Syst. Dyn.}, vol.~54, no.~8, pp. 1147--1176, 2016.

\bibitem{treiber2013traffic}
M.~Treiber and A.~Kesting, ``Traffic flow dynamics,'' \emph{Traffic Flow
  Dynamics: Data, Models and Simulation, Springer-Verlag Berlin Heidelberg},
  2013.

\bibitem{di2019cooperative}
M.~Di~Vaio, G.~Fiengo \emph{et~al.}, ``Cooperative shock waves mitigation in
  mixed traffic flow environment,'' \emph{IEEE Trans. Intell. Transp. Syst.},
  vol.~20, no.~12, pp. 4339--4353, 2019.

\bibitem{jin2014dynamics}
I.~G. Jin and G.~Orosz, ``Dynamics of connected vehicle systems with delayed
  acceleration feedback,'' \emph{Transp. Res. Part C Emerg. Technol.}, vol.~46,
  pp. 46--64, 2014.

\bibitem{zhou2020stabilizing}
Y.~Zhou, S.~Ahn \emph{et~al.}, ``Stabilizing mixed vehicular platoons with
  connected automated vehicles: An h-infinity approach,'' \emph{Transp. Res.
  Part B: Methodol.}, vol. 132, pp. 152--170, 2020.

\bibitem{sugiyama2008traffic}
Y.~Sugiyama, M.~Fukui \emph{et~al.}, ``Traffic jams without
  bottlenecks--experimental evidence for the physical mechanism of the
  formation of a jam,'' \emph{New J. Phys.}, vol.~10, no.~3, p. 033001, 2008.

\bibitem{stern2018dissipation}
R.~E. Stern, S.~Cui \emph{et~al.}, ``Dissipation of stop-and-go waves via
  control of autonomous vehicles: Field experiments,'' \emph{Transp. Res. Part
  C Emerg. Technol.}, vol.~89, pp. 205--221, 2018.

\bibitem{zheng2020smoothing}
Y.~Zheng, J.~Wang, and K.~Li, ``Smoothing traffic flow via control of
  autonomous vehicles,'' \emph{IEEE Internet Things J.}, vol.~7, no.~5, pp.
  3882--3896, 2020.

\bibitem{wang2020controllability}
J.~Wang, Y.~Zheng \emph{et~al.}, ``Controllability analysis and optimal control
  of mixed traffic flow with human-driven and autonomous vehicles,'' \emph{IEEE
  Trans. Intell. Transp. Syst.}, pp. 1--15, 2020.

\bibitem{cui2017stabilizing}
S.~Cui, B.~Seibold \emph{et~al.}, ``Stabilizing traffic flow via a single
  autonomous vehicle: Possibilities and limitations,'' in \emph{2017 IEEE IV},
  pp. 1336--1341.

\bibitem{piacentini2018traffic}
G.~Piacentini, P.~Goatin, and A.~Ferrara, ``Traffic control via moving
  bottleneck of coordinated vehicles,'' \emph{IFAC-PapersOnLine}, vol.~51,
  no.~9, pp. 13 -- 18, 2018.

\bibitem{zheng2016stability}
Y.~Zheng, S.~E. Li, J.~Wang \emph{et~al.}, ``Stability and scalability of
  homogeneous vehicular platoon: Study on the influence of information flow
  topologies,'' \emph{IEEE Trans. Intell. Transp. Syst.}, vol.~17, no.~1, pp.
  14--26, 2016.

\bibitem{ding2013network}
L.~Ding, Q.-L. Han, and G.~Guo, ``Network-based leader-following consensus for
  distributed multi-agent systems,'' \emph{Automatica}, vol.~49, no.~7, pp.
  2281--2286, 2013.

\bibitem{ni2010leader}
W.~Ni and D.~Cheng, ``Leader-following consensus of multi-agent systems under
  fixed and switching topologies,'' \emph{Syst. Control Lett.}, vol.~59, no.
  3-4, pp. 209--217, 2010.

\bibitem{jafari2011leader}
S.~Jafari, A.~Ajorlou, and A.~G. Aghdam, ``Leader localization in multi-agent
  systems subject to failure: A graph-theoretic approach,'' \emph{Automatica},
  vol.~47, no.~8, pp. 1744--1750, 2011.

\bibitem{clark2013supermodular}
A.~Clark, L.~Bushnell, and R.~Poovendran, ``A supermodular optimization
  framework for leader selection under link noise in linear multi-agent
  systems,'' \emph{IEEE Trans. Automat. Contr.}, vol.~59, no.~2, pp. 283--296,
  2013.

\bibitem{bando1995dynamical}
M.~Bando, K.~Hasebe \emph{et~al.}, ``Dynamical model of traffic congestion and
  numerical simulation,'' \emph{Phys. Rev. E}, vol.~51, no.~2, p. 1035, 1995.

\bibitem{Treiber2000Congested}
M.~Treiber, A.~Hennecke, and D.~Helbing, ``Congested traffic states in
  empirical observations and microscopic simulations,'' \emph{Phys. Rev. E},
  vol.~62, no.~2, p. 1805, 2000.

\bibitem{jin2017optimal}
I.~G. Jin and G.~Orosz, ``Optimal control of connected vehicle systems with
  communication delay and driver reaction time,'' \emph{IEEE Trans. Intell.
  Transp. Syst.}, vol.~18, no.~8, pp. 2056--2070, 2017.

\bibitem{skogestad2007multivariable}
S.~Skogestad and I.~Postlethwaite, \emph{Multivariable feedback control:
  analysis and design}.\hskip 1em plus 0.5em minus 0.4em\relax New York: Wiley,
  2007, vol.~2.

\end{thebibliography}

\end{document}